\documentclass[aps,prresearch,twocolumn,superscriptaddress,nofootinbib]{revtex4-2}

\usepackage{amsmath,amssymb,amsthm}
\usepackage{graphicx}
\usepackage[hidelinks]{hyperref}

\newtheorem{theorem}{Theorem}
\newtheorem{lemma}{Lemma}
\newtheorem{proposition}{Proposition}
\newtheorem{corollary}{Corollary}
\newtheorem{conjecture}{Conjecture}

\newcommand{\kap}{\kappa}
\newcommand{\ip}[2]{\langle #1,\, #2\rangle}

\begin{document}

\title{Conservation capacity of local learning rules in physical networks}

\author{Bijaya Dangol}
\affiliation{Independent researcher}

\begin{abstract}
How many memories can a local learning rule protect, and what fixes the
number? Physical learning rules train resistive networks through local
measurements, and the standard rules conserve a mass-like function of
the conductances, a law whose general form was posed as an open problem
with the expectation that no useful solution theory exists. We answer
it, in the direction the expectation ran against, and the answer is a
capacity theory. The Tellegen identity
behind the conserved mass localizes: the feedback state is pinned to
zero at the inputs, so every sector of the circuit that the input
electrodes separate carries its own private conserved mass, by an
argument consuming only Kirchhoff's law and that boundary condition,
hence valid for arbitrary nonlinear branch laws. The number of
independent sector masses is a topological property of the circuit,
bounded by the number of output electrodes. An untrainable element can
destroy the mass of its own sector and of no other, which says where
fixed nonlinear components may be placed; per-edge learning rates
select which functionals are conserved and never how many; and adjoint
coupled learning, which clamps its outputs before measuring, drains
every output-carrying sector at a rate set by the squared output
multipliers. In linear circuits we then classify the identities: exact
rational computation over all 502 circuits on up to five vertices with
two inputs and one or two outputs, and hundreds of larger ones, yields
a closed-form count, proved on that family and conjectured in general,
with the sector statement a theorem at every circuit size. The
anticipated series-parallel polynomial laws appear as exactly the
series-class cubic differences, and the number the opening question
asks for is a budget: one designed mass per sector, the differences
the topology donates, and one broadcast scalar for each protected
functional beyond.
\end{abstract}

\maketitle

\section{Introduction}
\label{sec:intro}

A network of adjustable resistors can be trained to compute:
contrastive local rules compare two equilibrium states of the circuit
and adjust each conductance using only quantities measurable at that
edge \cite{stern2021supervised,dillavou2022demonstration,
stern2023learning}. Equilibrium propagation (EP) nudges the outputs
with a force proportional to the error \cite{scellier2017equilibrium};
coupled learning (CL) clamps them a small step toward their targets
\cite{stern2021supervised}. In the continuous-time small-nudge limit
both rules exactly conserve the conductance mass
$K=\tfrac12\sum_e\kap_e^2$, a law established in
Ref.~\cite{mcginnis2026conservation} and used there to prove
convergence; untrainable elements, not the law itself, fix what a
trained circuit remembers of its initialization
\cite{dangol2026memory}, and modulating each
edge's learning rate chooses which mass a sector holds invariant
(Sec.~\ref{sec:capacity}).

All of these results rest on one conserved quantity, and what else a
circuit could conserve is a design question before it is a mathematical
one. If a circuit supports several independent conserved functionals, a
rule could protect several stored quantities at once; if it supports
one, every local rule faces the same memory budget, and enlarging it
requires communication. The mathematical form of the question is posed
as an open problem in Ref.~\cite{mcginnis2026conservation}: for a
general functional the conservation condition is an overdetermined
partial differential equation which that reference expects to admit no
nontrivial solutions, while its preliminary calculations suggest that
series-parallel topologies nonetheless admit polynomial conservation
laws beyond the squared norm. The answer runs against the expectation,
and it has two parts, the first needing almost nothing. The
identity behind the conserved mass is Tellegen's, and Tellegen's theorem
is a statement about Kirchhoff's laws rather than about elements, so it
localizes to any part of the circuit that the input electrodes separate
and it survives arbitrary nonlinear branch laws
(Sec.~\ref{sec:tellegen}). That already fixes the count of sector
masses, caps it by the number of outputs, and says where fixed
nonlinear components may sit without destroying a stored memory. The
second part is a completeness statement, and there linearity is
essential: we classify, on the certified family, every identity
available to every local rule, which needs the states to be rational in
the conductances so that numerators are spanning-forest sums.

We answer the separable case completely for an explicit family of
circuits and give the matching construction. A bootstrap lemma converts
the conservation condition from a functional equation into finite
linear algebra: fixing all conductances but one and using the
Sherman--Morrison structure of the network response forces every
per-edge coefficient function to be a polynomial of degree at most two
in its own conductance. The identity space of a given circuit is then
the null space of a finite system, which we compute in exact rational
arithmetic, with no floating point anywhere, for every connected graph
on up to five vertices with every terminal placement up to isomorphism,
502 configured circuits, and for hundreds of larger random circuits.
The computed
dimension matches a closed-form count in every case, and the count
itself came out of the certificates richer than the families we set out
to confirm. That localization is what the certificates see: even
standard EP conserves the mass of each grounded sector separately, a
refinement invisible in the aggregate law of
Ref.~\cite{mcginnis2026conservation}. Beyond the sector identities the
only further dimensions are degeneracies with complete
characterizations: series classes, where shared currents in both probe
states make cubic-mass differences conserved, the polynomial laws
anticipated for series-parallel topologies in
Ref.~\cite{mcginnis2026conservation}; parallel classes, where shared
drops make linear-mass differences conserved; and feedback-dead edges,
which no rule measuring against the free state can move, their
conductances being trivially invariant. Hard thresholding \cite{chatterjee2025remembrance} occupies
the frozen branch of this classification; rate modulation
(Sec.~\ref{sec:capacity}) designs within the sector branch.

The capacity statement follows, and its core is proved
unconditionally. The scaling homogeneity of the two probes grades the
identity space by polynomial degree, separating the sector, series, and
parallel families into independent slices; the linear slice is then
identified exactly, at every circuit size, by a reduction to a
certified landing zone (Theorem~\ref{thm:blocks}). Consequently every
local contrastive rule conserves one designable mass per grounded
sector, however the rule is built, and never two independent masses on
the same sector. The grounding pattern, a boundary-condition choice,
thereby fixes the circuit's free memory budget. Beyond it, one
global broadcast scalar per additional invariant closes the gap, with
correction directions the hardware already measures: auxiliary
contrastive fields satisfy the same structural identity as the update
and so leave every designed mass untouched, and we verify two
simultaneous designed masses at machine precision. The construction is
delicate in two instructive ways. Projecting along the target
functionals' gradient directions, the textbook constrained-flow ansatz,
conserves the targets while destroying the designed mass. And a
single-output network admits no ladder step at all: its contrastive
field family has rank one in the data, so no independent correction
direction exists.

\section{Setup}
\label{sec:setup}

A circuit is a connected graph with $N$ nodes, $M$ edges, incidence
matrix $D\in\mathbb{R}^{N\times M}$ with columns $d_e$, and positive
conductances $\kap\in\mathbb{R}^M_{>0}$, giving the Laplacian
$L(\kap)=\sum_e\kap_e d_e d_e^{\mathsf T}$ and energy
$G(x;\kap)=\tfrac12\ip{x}{L(\kap)x}$. Input nodes are clamped through
$P^{\mathsf T}x=v$ and outputs are read through $Q^{\mathsf T}x$, with
inputs and outputs disjoint. Training compares two probe states of the
same circuit \cite{scellier2017equilibrium,stern2021supervised,
mcginnis2026coercivity}: the free state $x_0(\kap,v)$, the energy
minimizer with inputs clamped and outputs floating, and the nudge
response $y(\kap,r)$, the state with inputs clamped to zero and
currents $Qr$ injected at the outputs, where for EP $r$ is the output
error. Edge $e$ measures its own drops in the two probes,
\begin{equation}
\alpha_e = d_e^{\mathsf T}x_0(\kap,v), \qquad
\beta_e = d_e^{\mathsf T}y(\kap,r),
\end{equation}
and a \emph{local rule} is any update of the form
\begin{equation}
\dot\kap_e = h_e(\kap_e,\alpha_e,\beta_e),
\label{eq:localrule}
\end{equation}
with $h_e$ analytic. The rule is \emph{contrastive} if
$h_e(\kap,\alpha,0)=h_e(\kap,0,\beta)=0$: no update without both a
forward signal and a feedback signal. EP is the bilinear case
$h_e=\alpha_e\beta_e$; CL is the same bilinear form with the output
error premultiplied by the Dirichlet-to-Neumann map
\cite{mcginnis2026coercivity}, which changes the injection data and not
the layer structure, so the two share an identity space. The
rate-modulated rules of Sec.~\ref{sec:capacity} are
$h_e=\alpha_e\beta_e/w_e$.

A separable functional $Q(\kap)=\sum_e q_e(\kap_e)$ is conserved by the
rule when $\sum_e q_e'(\kap_e)\,h_e(\kap_e,\alpha_e,\beta_e)=0$ holds
identically over $\kap>0$ and all data $(v,r)$; conservation along
every training trajectory, for every task and batch, is equivalent to
this identity because batch updates are averages of single-datum
updates. Since $\alpha$ is linear in $v$ and $\beta$ is linear in $r$,
the Taylor layers of $h_e$ in $(\alpha,\beta)$ are homogeneous of
distinct bidegrees and must vanish separately:

\begin{lemma}[Layer separation]
\label{lem:layers}
The rule \eqref{eq:localrule} conserves $Q$ if and only if, for every
bidegree $(p,q)$ with $p,q\ge1$,
\begin{equation}
\sum_e u_e(\kap_e)\,\alpha_e^{\,p}\beta_e^{\,q} = 0
\quad\text{identically}, \qquad
u_e = q_e'\,h_e^{(p,q)},
\label{eq:layer}
\end{equation}
where $h_e^{(p,q)}(\kap)$ is the $(p,q)$ Taylor coefficient of $h_e$ at
zero signal.
\end{lemma}

The learning content of a contrastive rule sits at $(1,1)$: layers even
in $\beta$ cannot follow the sign of the error, so a rule that trains
has $h^{(1,1)}\neq0$, and the standard rules are exactly bilinear. The
capacity question is therefore the structure of the solution space of
\eqref{eq:layer}, the \emph{identity space} of the circuit at each
layer, and chiefly at $(1,1)$.

\section{Conservation is local, and needs no linear elements}
\label{sec:tellegen}

The conserved mass of a local rule is not a property of the circuit as a
whole. It decomposes over the parts the input electrodes separate, and the
decomposition rests on so little that it survives arbitrary nonlinear
elements. Call the \emph{sectors} of the circuit the connected components of
the network with the input nodes deleted, each branch assigned to the sector
holding its non-input endpoints and branches joining two inputs assigned to
no sector. We avoid calling the sectors blocks: that word standardly
names a maximal biconnected subgraph, as in Ref.~\cite{cel1997}.
Write $E$ for the trainable branches and $F$ for any branch
carrying no trainable parameter, with convex $C^1$ laws on both.

\begin{lemma}[Localized Tellegen]
\label{lem:tellegen}
Suppose the free state satisfies Kirchhoff's current law with injection
permitted only at input nodes, and the feedback state $y$ vanishes at the
input nodes. Then for every sector $B$,
\begin{equation}
\sum_{b\in B} i_b(x_0)\,\Delta_b y=0,
\label{eq:localtellegen}
\end{equation}
with $i_b(x_0)$ the free branch current and $\Delta_b y$ the feedback drop.
\end{lemma}

\begin{proof}
Regrouping by node, $\sum_{b\in B}i_b\Delta_b y=\sum_n y_n J_n$ with $J_n$
the net free current leaving $n$ through the branches of $B$. At an input
node $y_n=0$. At any other node of $B$ every incident branch lies in $B$,
because $B$ is a full component of the circuit minus the inputs, so $J_n$
is the total net current there and vanishes for want of an injection.
\end{proof}

Only Kirchhoff's law and the boundary condition enter. No constitutive
relation appears anywhere in the proof, which is why the consequence below
holds for nonlinear elements while the classification of
Sec.~\ref{sec:classification} does not.

\begin{theorem}[Sector conservation, arbitrary element laws]
\label{thm:blockcons}
Let the energy be separable in the parameters,
$G(x;\kap)=\sum_{e\in E}g_e(\kap_e)\tilde G_e(\Delta_e x)
+\sum_{f\in F}\Phi_f(\Delta_f x)$, and let $f_e$ be an antiderivative of
$g_e/g_e'$. Under equilibrium propagation or coupled learning with per-edge
rates $w_e>0$, $\dot\kap_e=u_e/w_e$, the sector mass
$Q^B=\sum_{e\in E\cap B}w_ef_e(\kap_e)$ obeys
\begin{equation}
\frac{d}{dt}Q^B=-\sum_{f\in F\cap B} i_f(x_0)\,\Delta_f y .
\label{eq:leak}
\end{equation}
In particular $Q^B$ is exactly conserved whenever every branch of $B$ is
trainable, for arbitrary element laws.
\end{theorem}

\begin{proof}
$\dot Q^B=\sum_{E\cap B}w_ef_e'\dot\kap_e=\sum_{E\cap B}f_e'u_e$, the
weights cancelling branch by branch. The field of either rule is
$u_e=g_e'(\kap_e)\tilde G_e'(\Delta_e x_0)\,\Delta_e y$, so
$f_e'u_e=i_e(x_0)\,\Delta_e y$, and Lemma~\ref{lem:tellegen} splits the
sector sum into its trainable and untrainable parts.
\end{proof}

Branches joining two inputs are frozen: both endpoints carry $y=0$, so
their field vanishes and they never move, their masses are constants,
and the sector masses account for all of the global mass of
Ref.~\cite{mcginnis2026conservation} that moves. For $F$ empty and $B$
the whole network, Theorem~\ref{thm:blockcons} is Lemma 3.2 of
Ref.~\cite{mcginnis2026conservation}; what is new here is the
localization, the leak formula \eqref{eq:leak}, and the count.

\begin{corollary}[Capacity and its ceiling]
\label{cor:ceiling}
Let $b$ count the sectors that are entirely trainable and carry field. Their
masses are $b$ functionals conserved simultaneously and independently,
since sectors are disjoint as branch sets and each sector's weights may be
chosen alone. If every branch has positive curvature at the free state then
\begin{equation}
b\;\le\;\#\{\text{output electrodes}\},
\end{equation}
because a sector holding no output has a feedback state solving a
homogeneous problem with zero boundary data, hence $y\equiv0$ and no field
on that sector, and the sectors are disjoint.
\end{corollary}

The count is therefore topological, fixed by where the electrodes sit
rather than by how the rule is built, and capped by the number of outputs.
It also settles where fixed nonlinear components may be placed: by
\eqref{eq:leak} an untrainable branch can destroy the conserved mass of
its own sector and of no other, so confining a circuit's diodes to
designated sectors leaves every remaining sector's memory exactly
protected.

The ceiling bounds the number of \emph{sector}
masses, not the number of independent conserved masses a circuit can
carry. Sectors are not the finest set on which the pairing localizes: when
the element laws carry offsets, so that $\tilde G_e$ is minimized away
from zero drop, a circuit can hold more independent exactly conserved
masses than it has outputs. A four-node example has inputs $\{0,1\}$,
output $\{3\}$, branches $(0,2)$, $(1,2)$ and a parallel pair $(2,3)$,
with laws $(\Delta-\ell_e)^2/2$ and $\ell=(0,0,1,-1)$; it carries two
independent conserved masses and one output. Corollary~\ref{cor:ceiling}
is unaffected, because those four branches form a single sector, but the
bound should not be read as a bound on conservation itself.

The exact statement of the localization is a separation
condition. Writing $\mathcal{F}$ for the space of free currents obeying
Kirchhoff's law with injection only at the inputs and
$\mathcal{T}=\mathcal{F}^{\perp}$ for the space of feedback drops
vanishing on the inputs, the sum $\sum_{b\in S}i_b\Delta_b y$ vanishes for
all $i\in\mathcal{F}$ and all $y$ if and only if $\mathcal{F}$ splits as
$(\mathcal{F}\cap\mathbb{R}^S)+(\mathcal{F}\cap\mathbb{R}^{E\setminus S})$,
that is, if and only if $S$ separates the matroid represented by
$\mathcal{F}$. Sectors are such sets, since a cycle cannot pass twice
through the identified input node, and so are the branches joining two
inputs, which become loops there. The corresponding statement for the
unconstrained network, where the answer is a union of biconnected
components, is given in Ref.~\cite{cel1997}; the two differ precisely because
the boundary condition at the inputs contracts them, and the difference is
visible already on the four-cycle $u_1,a,u_2,b$ with inputs $\{u_1,u_2\}$,
which is one biconnected component while our two branches through $a$ and
through $b$ localize separately. Because the argument uses only the
orthogonality of the two spaces and never the incidence matrix, it applies
unchanged to any system with a static-kinematic duality, in particular to
central-force spring networks, where the incidence structure is the
rigidity matrix.

\begin{proposition}[Rates choose the basis and never enlarge it]
\label{prop:basis}
For fixed $w$ the functional $\sum_e c_e f_e(\kap_e)$ is conserved
whenever $c/w$ is constant on each sector. In linear circuits
Theorem~\ref{thm:blocks} supplies the converse, so each entirely
trainable sector then contributes exactly one conserved functional up to
scale; with offset element laws the pairing can localize more finely, as
above, and a sector may carry several. In every case $b$ does not depend
on $w$: rate modulation selects which functionals are conserved, never
how many sectors carry one.
\end{proposition}

The number itself is specific to the circuit:
Corollary~\ref{cor:ceiling} ties it to the electrode pattern, and it
has no counterpart for an optimizer, which has no locality budget.

\begin{proposition}[The adjoint rule is the exception]
\label{prop:al}
Adjoint coupled learning clamps the outputs in the very state whose
currents enter \eqref{eq:localtellegen}, so its stationarity admits an
injection $\mu$ at the outputs there and Lemma~\ref{lem:tellegen} does
not apply.
Instead
\begin{equation}
\frac{d}{dt}Q^B=-\sum_{f\in F\cap B}i_f\,\Delta_f y
\;-\!\!\sum_{n\in O\cap B}\!\!\mu_n^2 ,
\end{equation}
so every sector carrying an output drains at a rate set by the squared
output multipliers it holds. Summing over sectors returns the global
dissipation law $\dot K=-2\Phi^*$, twice the rule's own loss
\cite{dangol2026memory}, of which this is the localization.
Here the sector $B$ comprises all its branches, including those dead
for the free state, not only the live subnetwork; the law fails on the
live subnetwork alone, because the omitted branches are exactly the ones
the adjoint rule moves.
\end{proposition}

Sector conservation is thus available to the two rules that measure against
the free state and to no rule that clamps its outputs first. The adjoint
rule is not left with nothing, however: what it loses is the sector budget,
while the series and parallel invariants of
Sec.~\ref{sec:classification} survive it. Their defining components
contain no terminal, so the extra output injection lies outside them and
the same Kirchhoff argument applies in both of the adjoint rule's probe
states.

\section{The identity space}
\label{sec:identity}

\subsection{Bootstrap: the space is finite dimensional}

The identity \eqref{eq:layer} quantifies over functions
$u_e:\mathbb{R}_{>0}\to\mathbb{R}$, one per edge. The network response
removes this freedom.

\begin{lemma}[Bootstrap]
\label{lem:bootstrap}
Let edge $e_0$ be live at layer $(p,q)$, meaning
$\alpha_{e_0}^{\,p}\beta_{e_0}^{\,q}$ is not identically zero. If the
functions $u_e$ satisfy \eqref{eq:layer}, then $u_{e_0}$ is a
polynomial of degree at most $p+q$. In particular, at the learning
layer every $u_e$ on a live edge is a quadratic polynomial in
$\kap_e$.
\end{lemma}

\begin{proof}
Fix all conductances except $z:=\kap_{e_0}$, so
$L(z)=L_0+z\,d_{e_0}d_{e_0}^{\mathsf T}$. Both probes solve the same
bordered system with matrix $K(z)=K_0+z\,\hat g\hat g^{\mathsf T}$,
where $\hat g$ embeds $d_{e_0}$ in the bordered coordinates. The
$z$ dependence of $K$ is rank one, so by Cramer's rule every drop is a
ratio of polynomials of degree at most one over the common denominator
$D(z)=\det K(z)$, itself of degree at most one. When $K_0$ is
invertible, Sherman--Morrison gives
$D(z)\propto1+\sigma z$ with
$\sigma=\hat g^{\mathsf T}K_0^{-1}\hat g>0$, drops
$d_e^{\mathsf T}x(z)=(a_e+b_ez)/(1+\sigma z)$, and a constant numerator
for the modified edge's own drop,
$d_{e_0}^{\mathsf T}x(z)=a_{e_0}/(1+\sigma z)$; when $e_0$ is a bridge
whose removal strands an input-free piece, $\det K_0=0$,
$D(z)\propto z$, and the same degree count holds with
$d_{e_0}^{\mathsf T}x(z)=a_{e_0}/(cz)$, again a constant numerator.
Multiplying \eqref{eq:layer} by $D(z)^{p+q}$ makes every term with
$e\neq e_0$ a polynomial in $z$ of degree at most $p+q$, with
coefficients independent of $z$, while the $e_0$ term becomes
$u_{e_0}(z)$ times a constant that is nonzero for generic data and
frozen conductances because $e_0$ is live. Solving for $u_{e_0}(z)$
exhibits it as a polynomial of degree at most $p+q$ whose coefficients
depend on the frozen variables; since $u_{e_0}$ is a function of $z$
alone, the coefficients are constants.
\end{proof}

\begin{corollary}[Decidability]
\label{cor:decidable}
For a fixed circuit the identity space at layer $(1,1)$ is the null
space of a linear system in the $3M$ unknowns
$u_e(\kap)=A_e+B_e\kap+C_e\kap^2$, each rational sample
$(\kap,v,r)$ contributing one exact linear constraint. Sampling only
shrinks the null space, so the exact null space of any sampled system
is a rigorous upper bound on the identity space, and the circuit's
conservation capacity is computable in rational arithmetic.
\end{corollary}

\subsection{The identity families}

Four mechanisms produce identities (Fig.~\ref{fig:classification}),
and each is an exact statement
about the circuit rather than about a particular rule. Call an edge
\emph{live} when both its drops can be nonzero, and define the
\emph{input-separated sectors} as the connected components of the live
subnetwork after splitting the graph at the input nodes.

\emph{Localized Tellegen.} The free state satisfies
$L(\kap)x_0=P\lambda$, the nudge response satisfies
$P^{\mathsf T}y=0$, and the global pairing vanishes:
\begin{equation}
\sum_e \kap_e\,\alpha_e\beta_e
 = \ip{y}{L(\kap)\,x_0} = \ip{P^{\mathsf T}y}{\lambda} = 0 .
\label{eq:tellegen}
\end{equation}
The pairing in fact vanishes sector by sector. For a sector $B$,
\begin{equation}
\sum_{e\in B} \kap_e\,\alpha_e\beta_e
 = \sum_{n} y_n\, I_n^B,
\end{equation}
with $I_n^B$ the net free current into node $n$ through $B$,
and every term dies: at input nodes $y_n=0$, and at any other node of
$B$ the free-state Kirchhoff law has no injection while every
current-carrying edge at $n$ belongs to $B$. An edge at a
non-input node that carries free current is alive in the free probe,
and it is then alive in the nudge probe as well, because its component
of the graph minus the inputs contains $B$'s output, so appending the
path from the output to the edge's input--input witness turns the
witness into an output-to-input path through the edge. A
free-current-carrying
edge at $n$ is therefore live, hence in $B$ by connectedness. Each sector thus
contributes the identity $u_e=s_B\kap_e$ on $B$ with its own constant:
$b$ sectors give a $b$-dimensional family. In rule terms, a bilinear
rule $h_e=g_e(\kap_e)\alpha_e\beta_e$ conserves the per-sector masses
$Q_g^B=\sum_{e\in B}\!\int\!\kap/g_e(\kap)\,d\kap$ separately; for
standard EP this says the conductance mass of every grounded sector is
conserved on its own, a refinement of the aggregate law of
Ref.~\cite{mcginnis2026conservation}, and for $g_e=1/w_e$ it gives the
designed masses of Theorem~\ref{thm:dichotomy}, one per sector.

\emph{Feedback-dead edges.} On every component of the graph minus the
input nodes that contains no output, the nudge response is harmonic
with zero boundary data and no injection, hence zero, and an edge
joining two inputs has both endpoints clamped; such edges have
$\beta_e\equiv0$. Dually, edges in pieces attached through a single
node with no input inside carry no free current, $\alpha_e\equiv0$. A
dead edge never moves under a contrastive rule that measures against the
free state, every function of its conductance is then conserved, and its
coefficient in \eqref{eq:layer} is unconstrained. The adjoint rule is
again the exception: its reference state clamps the outputs, so it drives
current through edges that are $\alpha$-dead and $\beta$-live and trains
them, and those coefficients are not free for it. Freezing chosen edges by fiat
\cite{chatterjee2025remembrance} manufactures exactly this branch.

\emph{Series classes.} Two live edges $e,f$ are \emph{series} when, in
each probe's current-carrying subgraph, deleting both leaves a
component that contains exactly one endpoint of each and no injecting
node: no input for the free probe, no terminal for the nudge. All of
that component's current then enters through $e$ and $f$, and Kirchhoff
over the component forces $\kap_e\alpha_e=\pm\kap_f\alpha_f$ and
likewise for $\beta$. The two probes' injection-free components are in
fact the same set of nodes, because an edge carries free current
exactly when it lies on a simple path between two inputs and nudge
current exactly when it lies on a simple output-to-input path with
input-free interior, so any path joining the two components' shared
endpoints is live in both probes and neither component can extend
beyond the other. The two signs therefore agree and cancel in the
product:
$\kap_e^2\alpha_e\beta_e=\kap_f^2\alpha_f\beta_f$, and
$u_e=c\kap_e^2$, $u_f=-c\kap_f^2$ is an identity. A non-terminal
node of degree two is the local special case, the injection-free
component being the node itself, but series partners need not be
adjacent: any two-edge cut isolating an injection-free region pairs its
crossing edges. Under EP the conserved functional is the cubic-mass
difference $\tfrac13(\kap_e^3-\kap_f^3)$: the higher-order polynomial
conservation laws anticipated for series-parallel topologies in
Ref.~\cite{mcginnis2026conservation} exist, and this is their form. A
series class of $\ell$ edges contributes $\ell-1$ dimensions. The
per-probe condition has teeth in both directions: an edge that is dead
at the learning layer but carries one probe's current spoils that
probe's sharing at a chain node, while a region containing an output
still supports free-current sharing, since outputs inject nothing in
the free state, and contributes to the sector structure instead.

\emph{Parallel classes.} Live edges sharing both endpoints share both
drops, so $u_e=c$, $u_f=-c$ is an identity and EP conserves the
difference $\kap_e-\kap_f$; a class of $\ell$ edges contributes
$\ell-1$ dimensions. Only coincident edges qualify: a series chain
whose composite is parallel to an edge produces cross terms that are
not separable, and no identity.

\begin{figure*}
\includegraphics[width=\textwidth]{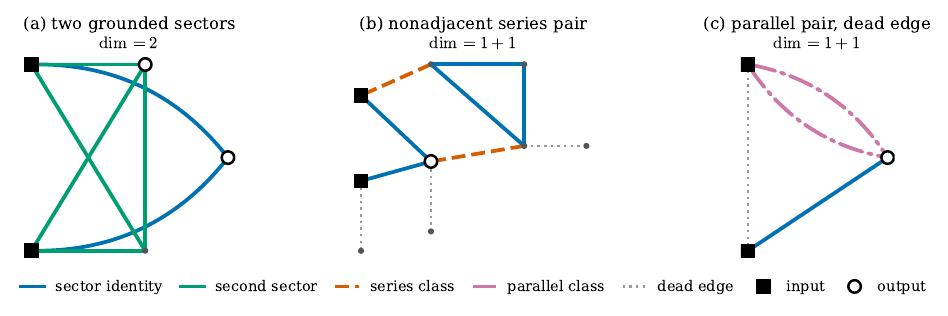}
\caption{The identity families on three example circuits (filled
squares: inputs; open circles: outputs). (a) A five-node circuit whose
live subnetwork splits at the inputs into two sectors, each carrying a
private Tellegen identity: even standard EP conserves the two sector
masses separately. (b) The two series-class edges (dashed) form a
two-edge cut isolating a terminal-free region and share current in
both probes without being adjacent; dotted edges are feedback-dead.
(c) A parallel pair (dash-dotted) shares both drops; the edge joining
the two inputs is dead. Each panel's dimension follows
Eq.~(\ref{eq:classification}).}
\label{fig:classification}
\end{figure*}

\section{Classification by exact certificate}
\label{sec:classification}

The families above are the lower bound; Corollary~\ref{cor:decidable}
supplies the upper bound circuit by circuit. For each circuit we draw
rational conductances and data, build the exact constraint rows over
the field of fractions, and compute the null space by exact Gaussian
elimination. The entire path from Kirchhoff solve to null-space
dimension uses rational arithmetic only, so the output is an integer
free of tolerance choices, and it reproduces identically on any
platform.

We ran the certificate on every connected simple graph on three, four,
and five vertices, with every placement of two input nodes and one or
two output nodes, counted up to isomorphism of the graph together with
its terminal sets, 502 terminal-configured circuits; on 300 random
circuits with six to nine nodes and random terminal placements; and on
named larger examples including wheels, complete graphs, a multigraph
with a parallel pair, and a twelve-edge random circuit. In every
configuration the exact dimension equals
\begin{equation}
\dim \;=\; b \;+\!\! \sum_{\text{series classes}}\!\!(\ell-1)
\;+\!\!\sum_{\text{parallel classes}}\!\!(\ell-1),
\label{eq:classification}
\end{equation}
where $b$ is the number of input-separated live sectors, with dead edges
contributing their unconstrained coefficients. No further identity
exists in any tested circuit. The enumerated family carries at most two
outputs, so it exhibits $b\le2$; as a witness beyond it, a ten-node
circuit of four disjoint lobes between the inputs, one output per lobe,
certifies $\dim=b=4$ exactly, with all four sector masses conserved
simultaneously in dynamics while the conductances move by up to $1.2$
times their initial values.

\begin{theorem}[Certified classification]
\label{thm:classification}
For every circuit in the certified family, the $(1,1)$ identity space
restricted to live edges is spanned by the sector Tellegen identities
and the series- and parallel-class families: its dimension is
\eqref{eq:classification}. In particular, when the live subnetwork is
input-connected and free of series and parallel degeneracy, the space
is one-dimensional.
\end{theorem}

\begin{conjecture}
\label{conj:general}
Equation~\eqref{eq:classification} is the $(1,1)$ identity-space
dimension of every circuit whose data span is not degenerate.
\end{conjecture}

Under a
global rescaling $\kap\to c\kap$ the free state is $0$-homogeneous and
the nudge response is $(-1)$-homogeneous, so evaluating a conservation
identity at $c\kap$ and collecting powers of $c$ makes each polynomial
degree of $u_e=A_e+B_e\kap_e+C_e\kap_e^2$ vanish separately: the
identity space is graded, and the constant, linear, and quadratic
slices are independent classification problems, hosting respectively
the parallel families, the sector identities, and the series families.
The per-degree certificate confirms the graded classification exactly,
$\dim_0=$ parallel classes, $\dim_1=b$, $\dim_2=$ series classes, on
the 400 of the 502 enumerated circuits with live edges and on 6182 of
6934 multigraphs, the remainder vacuous for want of a live edge, with
triple parallel classes included and zero mismatches. The linear slice is a theorem at every circuit size
(Sec.~\ref{sec:blocktheorem}); the constant and quadratic slices
remain conjectural, with the deletion-contraction route mapped and the
inhomogeneous terms priced by certified collapse criteria. The
degeneracy that marks the conjecture's boundary is visible in the
certificates: on very small live subnetworks, where the drop vectors
have rank one in the data, accidental identities appear that the
genericity clause excludes. A floating-point survey of the higher
layers, with singular-value gaps of $10^{10}$ and above separating the
identified null spaces from the rest of the spectrum, finds no global
identity at $(2,1)$, $(1,2)$, or $(2,2)$ on generic circuits, and on
degenerate ones the series- and parallel-class families together with
further accidental identities of the same rank-deficient kind.

\section{The sector theorem}
\label{sec:blocktheorem}

Theorem~\ref{thm:blockcons} says the sector masses \emph{are} conserved,
for arbitrary element laws, which bounds a circuit's capacity from below.
The theorem of this section is the converse half and is linear: in the
linear slice of the identity space there is nothing \emph{but} the sector
coefficients, so the lower bound is tight.

\begin{theorem}[Sectors]
\label{thm:blocks}
For every circuit, the linear slice of the $(1,1)$ identity space,
$u_e=w_e\kap_e$ on live edges, consists exactly of the coefficient
vectors constant on each input-separated live sector: its dimension is
$b$.
\end{theorem}

The converse inclusion is the localized Tellegen identity of
Sec.~\ref{sec:identity}. The forward inclusion is proved by a
reduction whose finite core is certified exactly. Linear-slice
identities restrict cleanly to minors: in the $z=\kap_g$ expansion the
deleted edge's term is $w_g z\,\alpha_g\beta_g$, which vanishes in the
$z\to0$ limit and is $O(1/z)$ in the $z\to\infty$ limit, so both the
deletion and the contraction of any edge carry the identity to the
smaller circuit with the same coefficients, with no inhomogeneous
term. Pendant input-free structures need never be removed: passive,
they conduct nothing, so they ride along as decorations that only add
constraints. For any two live edges sharing a non-input vertex, the
witness characterizations (an edge is free-live exactly when it lies
on a simple input-to-input path, and nudge-live exactly when it lies
on a simple output-to-input path with input-free interior, since no
nudge current passes a grounded node) supply access paths whose guided
contractions and deletions shrink the circuit to at most seven
vertices with both edges still live. Every possible landing
configuration is certified: over the 376\,978 circuits comprising all
configurations on up to six vertices carrying at most the three inputs
and two outputs that a reduced core can hold, and all seven-vertex
configurations with at most eight edges that contain the marked pair,
with liveness decided structurally and all arithmetic exact, every one
of the 253\,854 configurations with the pair live forces equal linear
coefficients on the pair, with zero violations. Equality propagates
along shared non-input vertices, which is exactly sector connectivity,
and Theorem~\ref{thm:blocks} follows.

\section{The capacity of a local rule}
\label{sec:capacity}

Theorem~\ref{thm:blocks} converts into a statement about every local
rule at once, at every circuit size. Designing retention through
weighted masses has antecedents outside physical learning:
designer-weighted consolidation variables \cite{benna2016computational},
per-parameter learning rates as the protection mechanism
\cite{ebrahimi2020uncertainty}, and preconditioners chosen to preserve
the flow's fixed points \cite{kao2021natural}. What the circuit adds is
that the invariant is exact and its count is topological.

\begin{theorem}[Designed-mass capacity]
\label{thm:dichotomy}
In a linear circuit, let the rule \eqref{eq:localrule} be contrastive
with bilinear
coefficient $g_e:=h_e^{(1,1)}$ not identically zero on any plastic
edge. On any circuit with $b$ input-separated live sectors, the
conserved functionals of weighted-mass type, $q_e'=w_e\kap_e/g_e$,
form exactly the space
\begin{equation}
\bigoplus_{B}\bigl\{\,s\,Q_g^B : s\in\mathbb{R}\,\bigr\},
\end{equation}
where $Q_g^B$ is the rule's designed mass on sector $B$, given by
$q_e'=\kap/g_e$ there. The designable conservation capacity of every
local contrastive rule equals the number of grounded sectors: one
designed mass per sector, and never two independent ones on the same
sector.
\end{theorem}

\begin{proof}
Conservation of $Q$ places $(q_e'g_e)_e$ in the $(1,1)$ identity
space; for weighted-mass functionals this vector is $w_e\kap_e$, in
the linear slice, which Theorem~\ref{thm:blocks} identifies as the
sector span: $w_e=s_B$ on each sector, one free constant per sector.
Purely bilinear rules attain the bound because
$\dot Q_g^B=\sum_{e\in B}\kap_e\alpha_e\beta_e=0$ by the localized
Tellegen identity.
\end{proof}

Under Conjecture~\ref{conj:general}, and unconditionally on every
certified circuit, the statement is complete: the conserved separable
functionals of any such rule are the per-sector designed masses, the
topology-donated series and parallel differences, and arbitrary
functionals of frozen-edge conductances, and nothing else, since the
constant and quadratic slices of $(q_e'g_e)_e$ are then confined to
the parallel and series families.

Every local contrastive rule that measures against the free state
therefore faces the same budget: one designable invariant per grounded
sector, plus the topology-donated series and parallel differences and
the trivial invariants on frozen edges. A rule
that clamps its outputs before measuring loses the whole sector budget, by
Proposition~\ref{prop:al}, and keeps every topology-donated series and
parallel difference: its identity space is a proper subspace of the
others'. The budget is otherwise a boundary-condition
choice. Splitting the input set so that the live circuit separates into
more sectors buys free invariants with zero communication, a design
lever that costs only electrode placement; on an input-connected
circuit the budget is one, however the rule is built. The threshold
rule of Ref.~\cite{chatterjee2025remembrance} sits entirely on the
frozen branch, protecting by not learning; rate modulation
(Theorem~\ref{thm:dichotomy}) designs the plastic branch, choosing
which mass each sector will hold fixed. What no local rule can do is
hold two independent designed masses on the same sector. A general rule
with layers beyond $(1,1)$ obeys the same bound, since each additional
layer only adds constraints; bilinearity is what makes the bound
attained.

\section{The broadcast ladder}
\label{sec:ladder}

One global scalar per additional invariant closes the gap beyond the
sector budget, and the correction directions come from measurements the
hardware already makes.

\begin{proposition}[Ladder]
\label{prop:ladder}
Let $u$ be the update field, let $u'_1,\dots,u'_{m-1}$ be auxiliary
contrastive fields, update fields measured under further boundary
conditions, so that each satisfies $\sum_e\kap_eu'_{j,e}=0$ by
\eqref{eq:tellegen}, and let $W=\mathrm{diag}(w)$ carry the designed
mass $Q_w$. Whenever the $(m{-}1)\times(m{-}1)$ system
$\sum_j\ip{a^i\!\circ\!\kap}{W^{-1}u'_j}\,c_j
=\ip{a^i\!\circ\!\kap}{W^{-1}u}$ is solvable, the corrected rule
\begin{equation}
\dot\kap = W^{-1}\Bigl(u-\textstyle\sum_{j=1}^{m-1}c_j\,u'_j\Bigr)
\label{eq:ladder}
\end{equation}
exactly conserves the $m$ functionals
$Q_w,Q_{a^1},\dots,Q_{a^{m-1}}$ at the price of $m-1$ broadcast scalars
per update.
\end{proposition}

\begin{proof}
Every correction is itself a contrastive field, so
$\dot Q_w=\sum_e\kap_e(u_e-\sum_jc_ju'_{j,e})=0$ term by term through
\eqref{eq:tellegen}, independent of the $c_j$; the remaining
derivatives vanish by the choice of the $c_j$.
\end{proof}

In simulation the two-mass ladder holds both invariants at the
$10^{-10}$ level and below under conservation-controlled integration,
reaching machine precision on smooth instances, while the no-broadcast
control drifts at the $10^{-2}$ to $10^{-1}$ level
(Fig.~\ref{fig:ladder}). Of twenty-four runs over three circuits,
twelve are recorded; the other twelve were discarded when a conductance
reached zero under the conservation-controlled integrator. Two of the
recorded runs held both invariants over the full horizon, two reached a
fixed point of the corrected flow, and eight ended when the step size
collapsed against the conservation tolerance and are reported up to
that time; the no-broadcast control survived on two.
Fig.~\ref{fig:ladder} traces one recorded instance, the two-output
wheel, at fixed step size to $t=50$.

\begin{figure}[t]
\includegraphics[width=\columnwidth]{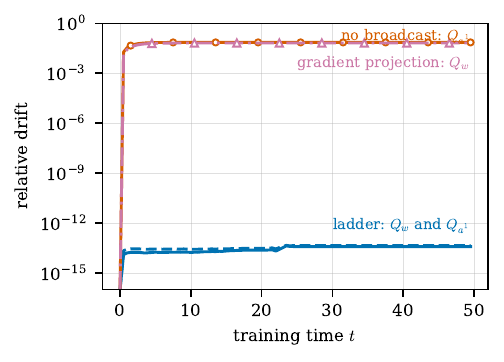}
\caption{Relative drift of the two designed masses on a two-output
wheel circuit under the corrected ladder (both masses at machine
precision, one broadcast scalar), the no-broadcast control (the
second mass drifts), and the gradient-direction projection (its
target holds while the designed mass $Q_w$ breaks). Fixed-step
integration; markers distinguish the two drifting curves.}
\label{fig:ladder}
\end{figure}

\emph{Corrections must be contrastive fields.} Projecting along the
gradient directions $W^{-1}(a\!\circ\!\kap)$ of the target functionals,
the textbook constrained-flow ansatz, conserves the target while
destroying the designed mass, which drifts at the $10^{-2}$ to
$10^{-1}$ level across our instances: generic directions do not satisfy
\eqref{eq:tellegen}, and $Q_w$'s conservation rests on it. The physics
supplies the correct correction cone for free; geometry does not.

\emph{Single-output circuits have no ladder.} For one output and two
input terminals the free drops have a fixed direction set by the input
difference and the nudge drops a fixed direction set by the single
injection, so the contrastive field family has rank one in the data and
every auxiliary field is parallel to $u$: the system in
Proposition~\ref{prop:ladder} is singular and \eqref{eq:ladder}
annihilates the update. Ladder steps exist once the field family has
rank two, requiring at least two outputs or three input terminals. The
single-output degeneracy parallels the trajectory equivalence of EP and
CL at one output \cite{dangol2026memory}: circuits with one output are
blind to distinctions that larger boundaries expose.

\emph{Descent is not free.} The corrected field \eqref{eq:ladder} is no
longer a metric gradient of the loss, and in our runs the training loss
sometimes rises where the uncorrected flow would descend. The full
metric projection onto the joint constraint set restores guaranteed
descent at the price of $m$ broadcasts rather than $m-1$. Whether the
gap is fundamental, one broadcast buying conservation and the next
buying conservation with descent, is open.

\section{Discussion}
\label{sec:discussion}

The capacity question has a quantifier that distinguishes it from
existing counting results. The conservation laws of a fixed flow, a
given architecture under Euclidean gradient descent, are counted by the
Lie-algebra rank machinery of Ref.~\cite{marcotte2023abide}, with
extensions to momentum and some non-Euclidean geometries
\cite{marcotte2024keep}; symmetry-derived laws in network training are
catalogued in Refs.~\cite{kunin2021neural,zhao2025symmetry}. Those
results characterize what a given dynamics happens to conserve. The
present bound quantifies over the rules: no local contrastive dynamics
on a certified input-connected circuit free of series and parallel
degeneracy conserves two independent
separable functionals on the same sector, whatever its form. The
balancedness invariants of deep linear networks \cite{du2018balanced},
a multi-parameter family of conserved norm differences, show that such
capacity is a property of the parametrization and its response
structure rather than of locality alone; in circuit terms they are
series-class invariants, and the degeneracy hypotheses are what exclude
them.

A quantifier also separates the combinatorial half of the results from the
nearest prior art on identities among spanning-forest polynomials.
Ref.~\cite{vlasev2012four} proves a quadratic forest identity on four
marked vertices and conjectures the dimension of the space of such
identities, and Ref.~\cite{fraser2023column} settles that conjecture:
the space of bilinear combinations vanishing on \emph{every} graph has
dimension $m(m-2)$ for $m$ marked vertices. Their
coefficients are indexed by set partitions of the marked vertices, a
finite index set independent of the graph, and their identities are
required to hold universally. Ours are indexed by the vertex pairs of one
fixed circuit, marked or not, and are required to hold identically in
that circuit's conductances. The two settings overlap only for two inputs
and two outputs, and there the universal solution space on the shared
index set is one-dimensional, spanned by the input-input pair whose
feedback drop vanishes identically, while a four-vertex path with
interior inputs has a four-dimensional identity space; a hand
computation with the two bases finds the relation spaces meeting in
zero. For more marked vertices the shapes cease to match
at all, their second factor being an $(m-1)$-forest sum where ours is
always a three-forest sum. The universal quantifier removes exactly the
graph-specific content that the capacity question is about.

The ladder's single-invariant rung is old practice. Adaptive flow
networks enforce a designed separable constraint through one global
multiplier \cite{hu2013adaptation,ronellenfitsch2016global}, and dual
decomposition broadcasts one scalar per constraint as a matter of
course. What was missing is the converse, and
Theorem~\ref{thm:dichotomy} supplies it: the multiplier is not an
implementation convenience but a necessity, since no local rule
achieves a second invariant per sector at zero communication. In the
language of the learning channel \cite{baldi2016local}, exact
conservation beyond the sector budget requires strictly positive
feedback bandwidth, at a rate of one scalar per functional.

In machine-learning vocabulary, a resistive network trained by
equilibrium propagation is an energy-based model whose energy function
is physical; the results here are the conservation structure of
contrastive training for that class, and they rest on the flow-effort
orthogonality of the network, which a general energy-based model does
not have.

For hardware the classification is an operating manual. A designer gets
one exact invariant per grounded sector at no communication cost and
chooses each through the rate profile of Theorem~\ref{thm:dichotomy};
splitting the input electrodes to multiply the sectors multiplies the
free invariants; series and parallel motifs and dead subnetworks add
protected quantities without any rule change, and the certificate
computes a given board's capacity outright; beyond that, each protected
functional costs one global feedback line, with corrections that the
training procedure already measures. On the certified family, retention
capacity, counted in independently protected separable functionals, is
the grounded sector count plus the broadcast budget, plus whatever the
topology donates.

Three boundaries hold. The classification theorem
is certified for the enumerated family and conjectured in general; the
deletion-contraction route outlined in Sec.~\ref{sec:classification} is
the intended proof, and the certificates are its base cases. The
bootstrap argument uses the rational dependence of linear-circuit
responses on each conductance, so nonlinear elements, whose responses
are not rational in the modified parameter, need separate treatment;
and since it is untrainable elements, not nonlinear ones, that break the
standard conservation law
\cite{mcginnis2026conservation,dangol2026memory}, while
Theorem~\ref{thm:blockcons} survives arbitrary trainable element laws,
the linear case is where the completeness question is sharply posed. And the functional class
is separable sums, the class closed under the locality of the rules
themselves; non-separable invariants of local rules, if any exist, are
outside the present counting.

\emph{Reproducibility.} All certificates, layer surveys, ladder runs,
and the enumeration are regenerated by the scripts in the accompanying
repository \cite{repo}, with the certificate path in exact rational
arithmetic end to end.




\begin{thebibliography}{99}

\bibitem{stern2021supervised}
M. Stern, D. Hexner, J. W. Rocks, and A. J. Liu,
Supervised learning in physical networks: From machine learning to
learning machines,
Phys. Rev. X \textbf{11}, 021045 (2021).

\bibitem{dillavou2022demonstration}
S. Dillavou, M. Stern, A. J. Liu, and D. J. Durian,
Demonstration of decentralized physics-driven learning,
Phys. Rev. Appl. \textbf{18}, 014040 (2022).

\bibitem{stern2023learning}
M. Stern and A. Murugan,
Learning without neurons in physical systems,
Annu. Rev. Condens. Matter Phys. \textbf{14}, 417 (2023).

\bibitem{scellier2017equilibrium}
B. Scellier and Y. Bengio,
Equilibrium propagation: Bridging the gap between energy-based models
and backpropagation,
Front. Comput. Neurosci. \textbf{11}, 24 (2017).

\bibitem{mcginnis2026conservation}
J. A. McGinnis, A. G. Kline, and Y. Mori,
A conservation law for equilibrium propagation and coupled learning,
arXiv:2606.15444.

\bibitem{dangol2026memory}
B. Dangol,
Untrainable elements determine what physical
learning remembers,
arXiv:2608.00097.

\bibitem{chatterjee2025remembrance}
P. Chatterjee, M. Guzman, and A. J. Liu,
Remembrance of tasks past in tunable physical networks,
arXiv:2512.03799.

\bibitem{mcginnis2026coercivity}
J. A. McGinnis, X. Li, and Y. Mori,
Coercivity and local convergence of physical learning in linear
circuits,
arXiv:2606.15443.

\bibitem{cel1997}
J.~Cel,
\emph{Tellegen's theorem for subnetworks},
J. Circuits Syst. Comput. \textbf{7}, 641 (1997).

\bibitem{benna2016computational}
M.~K. Benna and S. Fusi,
\emph{Computational principles of synaptic memory consolidation},
Nat. Neurosci. \textbf{19}, 1697 (2016).

\bibitem{ebrahimi2020uncertainty}
S. Ebrahimi, M. Elhoseiny, T. Darrell, and M. Rohrbach,
\emph{Uncertainty-guided continual learning with Bayesian neural
networks}, International Conference on Learning Representations (2020).

\bibitem{kao2021natural}
T.-C. Kao, K.~T. Jensen, G.~M. van de Ven, A. Bernacchia, and
G. Hennequin,
\emph{Natural continual learning: success is a journey, not (just) a
destination}, Advances in Neural Information Processing Systems
\textbf{34} (2021).

\bibitem{marcotte2023abide}
S. Marcotte, R. Gribonval, and G. Peyr\'e,
Abide by the law and follow the flow: Conservation laws for gradient
flows,
Advances in Neural Information Processing Systems \textbf{36} (2023).

\bibitem{marcotte2024keep}
S. Marcotte, R. Gribonval, and G. Peyr\'e,
Keep the momentum: Conservation laws beyond Euclidean gradient flows,
Proceedings of the 41st International Conference on Machine Learning
(2024).

\bibitem{kunin2021neural}
D. Kunin, J. Sagastuy-Brena, S. Ganguli, D. L. K. Yamins, and H. Tanaka,
Neural mechanics: Symmetry and broken conservation laws in deep
learning dynamics,
International Conference on Learning Representations (2021).

\bibitem{zhao2025symmetry}
B. Zhao, R. Walters, and R. Yu,
Symmetry in neural network parameter spaces,
Transactions on Machine Learning Research (2026).

\bibitem{du2018balanced}
S. S. Du, W. Hu, and J. D. Lee,
Algorithmic regularization in learning deep homogeneous models: Layers
are automatically balanced,
Advances in Neural Information Processing Systems \textbf{31} (2018).

\bibitem{vlasev2012four}
A.~Vlasev and K.~Yeats,
\emph{A four-vertex, quadratic, spanning forest polynomial identity},
Electron. J. Linear Algebra \textbf{23}, 923 (2012).

\bibitem{fraser2023column}
M.~Fraser and K.~Yeats,
\emph{Column expansion identities and quadratic spanning forest
identities}, Australas. J. Combin. \textbf{86}, 271 (2023).

\bibitem{hu2013adaptation}
D. Hu and D. Cai,
Adaptation and optimization of biological transport networks,
Phys. Rev. Lett. \textbf{111}, 138701 (2013).

\bibitem{ronellenfitsch2016global}
H. Ronellenfitsch and E. Katifori,
Global optimization, local adaptation, and the role of growth in
distribution networks,
Phys. Rev. Lett. \textbf{117}, 138301 (2016).

\bibitem{baldi2016local}
P. Baldi and P. Sadowski,
A theory of local learning, the learning channel, and the optimality of
backpropagation,
Neural Netw. \textbf{83}, 51 (2016).

\bibitem{repo}
Code, certificates, and experiments:
\url{https://github.com/dangoldbj/physical-learning-capacity}.

\end{thebibliography}
\end{document}